\documentclass[12pt]{article}

\usepackage{amsmath}
\usepackage{amssymb}
\usepackage{amsthm}
\usepackage{url}
\usepackage{graphicx}
\usepackage{natbib}
\usepackage[top=1in, bottom=1in, left=1in, right=1in]{geometry}
\usepackage{multirow}
\usepackage{cases}
\usepackage{mathtools}

\newtheorem{theorem}{Theorem}
\newtheorem{lemma}[theorem]{Lemma}
\newtheorem{proposition}[theorem]{Proposition}

\newtheorem{corollary}[theorem]{Corollary}
\newtheorem{definition}[theorem]{Definition}
\newtheorem{remark}[theorem]{Remark}
\newtheorem{example}[theorem]{Example}
\newtheorem{claim}[theorem]{Claim}
\newtheorem{fact}[theorem]{Fact}

\newtheorem{conjecture}[theorem]{Conjecture}
\newtheorem{aside}[theorem]{aside}

\ifdefined \theorem
\else
\theoremstyle{definition}
\newtheorem{theorem}{Theorem}
\newtheorem{lemma}[theorem]{Lemma}
\newtheorem{proposition}[theorem]{Proposition}
\newtheorem{corollary}[theorem]{Corollary}
\newtheorem{definition}[theorem]{Definition}
\newtheorem{remark}[theorem]{Remark}

\newtheorem{fact}[theorem]{Fact}

\newtheorem{conjecture}[theorem]{Conjecture}

\fi

\newcommand{\dom}{\text{dom }}

\newcommand{\mycomment}[1]{}

\newcommand{\argmax}[1]{\underset{#1}{\operatorname{argmax}}\;}
\newcommand{\argmin}[1]{\underset{#1}{\operatorname{argmin}}\;}

\newcommand{\one}{\bold{1}}

\newcommand{\bbE}{\mathbb{E}}

\newcommand{\bbN}{\mathbb{N}}

\newcommand{\bbR}{\mathbb{R}}

\newcommand{\mcD}{\mathcal{D}}
\newcommand{\mcE}{\mathcal{E}}

\newcommand{\mcU}{\mathcal{U}}

\newcommand{\ra}{\rightarrow}

\newcommand{\Polya}{P\'{o}lya}

\newcommand{\PG}{\text{PG}}
\newcommand{\Ga}{\text{Ga}}
\newcommand{\J}{J^*}
\newcommand{\JJ}{J^*}

\newcommand{\IGa}{\text{IGa}}
\newcommand{\IG}{\text{IG}}

\newcommand{\simiid}{\stackrel{iid}{\sim}}

\newcommand{\utanh}[1]{\frac{\tanh{\sqrt{#1}}}{\sqrt{#1}}}
\newcommand{\utan}[1]{\frac{\tan{\sqrt{#1}}}{\sqrt{#1}}}

\newcommand{\dd}[2]{\frac{d #1}{d #2}}
\newcommand{\eqindist}{\stackrel{\mcD}{=}}

\title{Sampling \Polya-Gamma random variates: alternate and approximate techniques}

\author{Jesse Windle, Nicholas G.\ Polson, James G.\ Scott}

\begin{document}

\maketitle

\abstract{Efficiently sampling from the \Polya-Gamma distribution, $\PG(b,z)$,
  is an essential element of \Polya-Gamma data augmentation.
  \cite{polson-etal-2013} show how to efficiently sample from the $\PG(1,z)$
  distribution.  We build two new samplers that offer improved performance when
  sampling from the $\PG(b,z)$ distribution and $b$ is not unity.

\tableofcontents

\section{Introduction}

Efficiently sampling \Polya-Gamma random variates is an essential element of the
eponymously named data augmentation technique \citep{polson-etal-2013}.  The
technique is applicable whenever one encounters a posterior distribution of the
form
\begin{equation}
\label{eqn:pg-posterior}
p(\beta | y) \propto p(\beta) \prod_{i=1}^n \frac{(e^{\psi_i})^{a_i}}{(1+e^{\psi_i})^{b_i}}
\end{equation}
where $\psi_i = x_i \beta$ and $a_i$ and $b_i$ are some functions of the data
$y$ and other parameters.  Introducing the auxiliary variables $\omega =
(\omega_i)_{i=1}^N$, independently distributed according to
\[
(\omega_i | \beta, y) \sim \PG(b_i, \psi_i)
\]
where $\PG(b_i, \psi_i)$ is a \Polya-Gamma random variate, yields the joint
density $p(\beta, \omega | y) \propto p(\omega | \beta, y) p(\beta | y)$ whose
complete conditional $p(\beta | \omega, y)$ is Gaussian.  Thus, one may
approximate the joint density by iteratively sampling from $p(\beta | \omega,
y)$ and $p(\omega | \beta, y) = \prod_{i=1}^N p(\omega_i | \psi_i, y)$.
Clearly, the effective sampling rate of this Markov Chain depends upon how
quickly one can sample \Polya-Gamma random variates.  (The effective sampling
rate is the rate at which a Markov Chain can produce approximately independent
samples.)  \cite{polson-etal-2013} showed how to efficiently sample from the
$\PG(1,z)$ distribution.  Here we consider alternative techniques for sampling
from the \Polya-Gamma distribution which are useful for other portions of its
parameter space.  We will construct an alternative sampler that is useful for
drawing $\PG(b, z)$ when $b \in \bbR^n$ is greater than unity, though not too
large, and an approximate sampler for drawing $\PG(b, z)$ when $b$ is large.
(This manuscript is a revised version of a chapter in the first author's
dissertation \citep{windle-thesis-2013}.)}

\section{The \Polya-Gamma distribution}
\label{sec:polya-gamma-distribution}


\begin{definition}[The \Polya-Gamma Distribution]
  Suppose $b > 0$ and $z \geq 0$.  The \Polya-Gamma distribution $\PG(b)$ is
  defined by the density $p_{PG}(x|b)$ on $\bbR^+$ with respect to Lebesgue measure
  that has the Laplace transform
  \[
  \cosh^{-b}(\sqrt{t/2}) = \int_{0}^\infty \exp(-t x) p_{PG}(x | b) dx.
  \]
  A random variable $X \sim \PG(b, z)$ for $z > 0$ is defined by exponentially
  tilting the $\PG(b)$ family:
  \[
  p_{PG}(x | b, z) = \cosh^{b}(z / 2) \exp(-x z^2 / 2) p_{PG}(x | b).
  \]
\end{definition}


We need to verify that this is, indeed, a valid Laplace transform.
\cite{biane-etal-2001} essentially show this and many other properties in their
survey of laws that connect analytic number theory and Brownian excursions.  One
of the laws surveyed, which we denote by $\JJ(b)$, has a Laplace transform given
by
\[
\bbE[ e^{-t \JJ(b)} ] = \cosh^{-b} (\sqrt{2t}).
\]
\cite{biane-etal-2001} show that this distribution has a density and derive
\emph{one} of its representations.  Thus, the existence of $\PG(b) = \JJ(b) / 4$
is verified and the definition of $\PG(b,z)$ is valid.  When devising samplers,
we find it convenient to work with the $\JJ(b)$ distribution since there is then
a trove of prior work to reference directly, instead of obliquely by a
re-scaling.  Similar to the definition of $\PG(b,z)$, we define $\JJ(b,z)$ by
exponential tilting:
\[
p_{\JJ}(x | z, b) = \cosh^{b}(z) e^{-x z^2 / 2} p_{\JJ}(x|b).
\]
Equivalently:

\begin{definition}
\label{def:jstar}
$\JJ(b,z)$ is the distribution with Laplace transform
\[
\cosh^b(z) \cosh^{-b}(\sqrt{2t+z^2}).
\]
\end{definition}

\begin{fact}
\label{fact:jstar}
The following aspects of the $\JJ(b,z)$ distribution are useful.
\begin{enumerate}
\item $\PG(b,z) = \frac{1}{4}\J(b, z/2)$.

\item $\JJ(b)$ has a density and it may be written as
\label{item:jstar-density-igamma}
\begin{displaymath}
p_{\JJ}(x|b) = \frac{2^b}{\Gamma(b)} \sum_{n=0}^{\infty} (-1)^n \frac{\Gamma(n+b)}{\Gamma(n+1)} 
\frac{(2n+b)}{\sqrt{2 \pi x^3}} \exp \Big( - \frac{(2n+b)^2}{2 x} \Big).
\end{displaymath}
Thus, the density of $\JJ(b,z)$ is
\[
p_{\JJ}(x|b,z) = \cosh^{b}(z) e^{-x z^2 / 2} p_{\JJ}(x|b).
\]

\item \label{item:jstar-convo}
The $\JJ(b,z)$ distribution is infinitely divisible.  Thus, if $X \sim
  \JJ(nb,z)$ where $b>0$ and $n \in \bbN$, and $X_i \simiid \JJ(b, z)$ for $i=1,
  \ldots, n$, then
  \[
  X \eqindist \sum_{i=1}^n X_i.
  \]

\item The moment generating function of $\JJ(b,z)$ is
  \[
  M(t; b, z) = \cosh^b(z) \cos^b(\sqrt{2t - z^2})
  \]
  and may be written as an infinite product
  \[
  \prod_{n=0}^\infty \Big(1 - \frac{t}{d_n}\Big)^{-b} \, , \; d_n = \frac{\pi^2}{2}
  \Big(n + \frac{1}{2}\Big)^2 + \frac{z^2}{2}.
  \]
  
\item Hence, $\JJ(b,z)$ is an infinite convolution of gammas and can be
  represented as
  \label{item:sum-of-gammas}
  \[
  \JJ(b,z) \sim \sum_{n=0}^\infty \frac{g_n}{d_n} \, , \; g_n \simiid \Ga(b,1).
  \]


\end{enumerate}
\end{fact}

\begin{proof}

  \cite{biane-etal-2001} provide justification for items (2), (3), and
  essentially (5).  Justification for items (1) and (4) are in
  \cite{polson-etal-2013}, though we present the arguments here.  For item (1),
  let $X = \J(b,z/2)$ and $Y = X / 4$ transform
  \[
  p_{\J}(x | b, z/2) dx = \cosh^b(z/2) \exp \Big( -\frac{x}{4} \frac{z^2}{2}
  \Big) p_{\J}(x | b) dx
  \]
  to
  \[
  \cosh^b(z/2) \exp \Big(-y \frac{z^2}{2} \Big) p_{\J}(4y | b) d(4y) =
  \cosh^b(z/2) \exp (- y z^2 /2 ) p_{PG}(y | b) dy.
  \]
  The last expression is by definition $Y \sim \PG(b, z)$.  Regarding (4),
  recall the Laplace transform of $\JJ(b,z)$ (Definition \ref{def:jstar}) is
  \[
  \varphi(t|b,z) = \cosh^{b}(z) \cosh^{-b}(\sqrt{2t + z^2}).
  \]
  By the Weierstrass factorization theorem \citep{pennisi-1976-book},
  $\cosh(\sqrt{2t})$ can be written as
  \[
  \cosh(\sqrt{2t}) = \prod_{n=0}^\infty \Big( 1 + \frac{t}{c_n} \Big) \, , \; 
  c_n =  \frac{\pi^2}{2} (n + 1/2)^2.
  \]
  Taking the reciprocal of $\varphi(t|1,z)$ yields
  \[
  \frac {\cosh(\sqrt{2t+z^2})}{\cosh (z)}
  =
  \frac {\prod_{n=0}^\infty \Big(1 + \frac{t + z^2/2}{c_n} \Big)}
  {\prod_{n=0}^\infty \Big(1 + \frac{z^2/2}{c_n} \Big)}
  = 
  \prod_{i=0}^\infty \Big(1 + \frac{t}{c_n + z^2/2}\Big);
  \]
  thus,
  \[
  \varphi(t|b,z) = \prod_{n=0}^\infty \Big(1 + \frac{t}{d_n}\Big)^{-b} \, , \; d_n
  = \frac{\pi^2}{2} (n + 1/2)^2 + \frac{z^2}{2}.
  \]
  Since $\varphi(-t;b,z) = M(t;b,z)$ we have
  \[
  M(t;b,z) = \prod_{n=0}^\infty \Big(1 - \frac{t}{d_n}\Big)^{-b}
  \]
  and
  \[
  \frac{M(t;b,z)}{\cosh^{b}(z)} = \cosh^{-b}(\sqrt{-2t + z^2}) =
  \cos^{-b}(\sqrt{2t - z^2}).
  \]
  Regarding item (5), one may invert the infinite product representation of
  Laplace transform to show that
  \[
  \JJ(b,z) \sim \sum_{n=0}^\infty \frac{g_n}{d_n} \, , g_n \simiid \Ga(b,1).
  \]

\end{proof}


Below we describe \cite{polson-etal-2013}'s $\JJ(1,z)$ sampler, which is
motivated by \cite{devroye-2009} and which relies on a reciprocal relationship
noticed by \cite{ciesielski-taylor-1962}, who show that in addition to Fact
(\ref{fact:jstar}.\ref{item:jstar-density-igamma}) one may represent the density
of a $\JJ(1)$ random variable as
\begin{equation}
\label{eqn:jstar-density-exp}
\sum_{n=0}^\infty (-1)^n \pi \Big(n+\frac{1}{2}\Big) e^{-(n+1/2)^2 \pi^2 x / 2}.
\end{equation}
By pasting these two densities together, one can construct an extremely
efficient sampler.  Unfortunately, there is no known general reciprocal
relationship that would extend this approach to $\JJ(n)$ for general $n$;
however, \cite{biane-etal-2001} provide an alternate density for the $\JJ(2)$
distribution based upon a reciprocal relationship with another random variable.

While there may not be an obvious reciprocal relationship to use, one may find
other alternate representations for the density of $\JJ(b)$ random variables
when $b$ is a positive integer.  Exploiting an idea from \cite{kent-1980} for
infinite convolutions of exponential random variables, one may invert the moment
generating function using partial fractions.  Consider the moment generating
function of $\JJ(h)$:
\begin{equation}
\label{eqn:mgf-product}
M(t) = \prod_{n=0}^\infty \Big(1 - \frac{t}{c_n}\Big)^{-h} \, , \; c_n =  \frac{\pi^2}{2} (n + 1/2)^2
\end{equation}
This can be expanded by partial fractions so that
\begin{equation}
\label{eqn:mgf-partial-fraction}
M(t) = \sum_{n=0}^\infty \sum_{m=1}^h \frac{A_{nm}}{(t - c_n)^m}.
\end{equation}
Inverting this sum term by term we find that one can represent the density as
\[
f(x|h) = \sum_{n=0}^\infty \sum_{m=1}^h A_{nm} \frac{x^{m-1} e^{-c_i x}}{(m-1)!},
\]
and infinite sum of gamma kernels.

To find formulas for the $\{A_{nm}\}_{nm}$ coefficients, consider the Laurent
series expansion of $M(t)$ about $c_i$.
\begin{equation}
\label{eqn:mgf-laurent}
M(t) = \sum_{n=0}^\infty a_{n}^{(i)} (t - c_i)^n + \sum_{m=1}^h
\frac{b_m^{(i)}}{(t - c_i)^m}.
\end{equation}
Such an expansions is valid since $c_n$ is an isolated singular point.  Since
the coefficients at the pole are unique, comparing coefficients in
(\ref{eqn:mgf-partial-fraction}) and (\ref{eqn:mgf-laurent}) shows that $A_{im}
= b_m^{(i)}$.  Further, one may calculate $b_m^{(i)}$ by considering the
function
\[
\nu_h(t) = (t - c_i)^h M(t)
\]
and then computing
\[
b_m^{(i)} = \frac{\nu_h^{(h-m)}(c_i)}{(h-m)!}.
\]
(See \cite{churchill-1984-book}.)  Writing the MGF in product form, as in
(\ref{eqn:mgf-product}), we see that
\[
\nu_h(t) = (-c_i)^h \prod_{n \neq i} \Big(1 - \frac{t}{c_n} \Big)^{-h}.
\]
Define
\[
\psi_h(t) = h \log (-c_i) - h \sum_{n \neq i} \Big(1 - \frac{t}{c_n}\Big).
\]
Then $\nu_h(t) = \exp \psi_h(t)$ and the derivatives of $\nu$ can then be
expressed as
\begin{align*}
\nu_h' & = e^{\psi_h} \psi_h'; \\
\nu_h'' & = e^{\psi_h} (\psi_h')^2 + e^{\psi_h} \psi_h''; \\
\nu_h''' & = e^{\psi_h} (\psi_h')^3 + 3 e^{\psi_h} \psi_h' \psi_h'' + e^{\psi_h}
\psi_h''' \\
\ldots & = \ldots \; \; 
\end{align*}
where
\[
\psi_1^{(k)}(t) = (k-1)! \sum_{n \neq i} ( c_n - t )^{-k}.
\]
Thus, one may calculate $b_{m}^{(i)}$ numerically using $\psi_h^{(k)}$, though
the convergence may be slow.

However, the most important coefficient, $b_h^{(i)}$, is already known.  Make
the dependence of $b_m^{(i)}$ on $h$ explicit by writing $b_m^{(i)}(h)$.  From
the formulas above we know that \( b_h^{(i)}(h) = \nu_h(c_i) \) and that \(
\nu_h(c_i) = \exp(\psi_1(c_i))^h.  \) But \( \exp(\psi_1(c_i)) = \nu_1(c_i) =
b_1^{(i)}(1).  \) From the reciprocal relationship provided at the start of the
section, we know that \( b_1^{(i)}(1) = (-1)^i \sqrt{2 c_i}.  \) Thus,
\[
A_{ih} = b_h^{(i)}(h) = (-1)^{ih} (2 c_i)^{h/2}.
\]
For $h \in \bbN$, the density for $\JJ(h)$ takes the form
\begin{equation}
\label{eqn:alternate-density}
f(x|h) = \sum_{n=0}^\infty
\Big[ \sum_{m=1}^h \frac{A_{nm} (h-1)! }{A_{nh} (m-1)!} \frac{1}{x^{h-m}} \Big]
\frac{A_{nh} x^{h-1} e^{-c_i x}}{(h-1)!}
\end{equation}
so the $A_{nh}$ terms dominate for large $x$.  Further, among those terms, the
first,
\[
\frac{A_{0h} x^{h-1} e^{-c_0 x}}{(h-1)!} = \frac{(\pi/2)^{h} x^{h-1} e^{-c_0 x}}{(h-1)!},
\]
should dominate as $x \ra \infty$.

\begin{remark}
\label{remark:tails}
This provides insight into the tail behavior of the $\JJ(h)$ distribution.  For
the right tail, we expect the density to decay as a $\Ga(h,c_0)$ distribution.
Examining the representation (\ref{fact:jstar}.\ref{item:jstar-density-igamma}),
we expect the left tail to decay like $\IGa(1/2, h^2/2)$.  These two
observations will prove useful when finding an approximation of the $\JJ(h)$
density.  We may multiply each of these densities by $e^{-xz^2/2}$ to determine
the tail behavior of $\JJ(h,z)$: the right tail should look like $\Ga(h, c_0 +
z^2/2)$ while the left tail should look like $\IG(\mu=h/z, h^2)$.
\end{remark}


\section{A $\JJ(n,z)$ sampler for $n \in \bbN$}
\label{sec:jstar1}

\cite{polson-etal-2013} show how to efficiently sample from the $\JJ(1,z)$
distribution.  Their approach is motivated by \cite{devroye-2009}.  This section
recaps that work, since it will help clarify the provenence of the other
samplers in this paper.

The $\J(1,z)$ sampler employs von Neumann's alternating sum method
\citep{devroye-1986}, which is an accept/reject algorithm for densities that may
be represented as infinite, alternating sums.  To remind the reader about
accept/reject samplers, one generates a random variable $Y$ with density $f$ by
repeatedly generating a proposal $X$ from density $g$ and $U$ from $\mcU(0, c \,
g(X))$ where $c \geq \|f/g\|_\infty$ until
\[
U \leq f(X); \; \text{ then set } Y \leftarrow X.
\]
(See \cite{robert-casella-2005-book} for more details.)  The von Neumann
alternating sum method requires that the density be expressed as an infinite,
alternating sum
\[
f(x) = \lim_{n \ra \infty} S_n(x), \; S_n(x) = \sum_{i=0}^n (-1)^i a_i(x)
\]
for which the partial sums $S_i$ satisfy the \emph{partial sum criterion}
\begin{equation}
\label{eqn:partial-sum-criterion}
\forall x, \; S_0(x) > S_2(x) > \ldots > f(x) > \ldots > S_3(x) > S_1(x),
\end{equation}
which is equivalent to the sequence $\{a_i(x)\}_{i=1}^\infty$ decreasing in $i$
for all $x$.  In that case, we have that $u < f(x)$ if and only if there is some
odd $i$ such that $u \leq S_i(x)$ and $u > f(x)$ if and only if there is some
even $i$ such that $u \geq S_i(x)$.  Thus one need not calculate the infinite
sum to see if $u < f(x)$, one only needs to calculate as many terms as necessary
to find that $u \leq S_i(x)$ for odd $i$ or $u \geq S_i(x)$ for even $i$.  (We
must be careful when $x=0$.)  One rarely needs to calculate a partial sum past
$S_1(x)$ before deciding to accept or reject \citep{polson-etal-2013}.

\subsection{Sampling from $\J(1,z)$}

The $\JJ(1)$ density may be represented in two different ways
\[
f(x) = \sum_{i=0}^n (-1)^n a_n^L(x) = \sum_{i=0}^n (-1)^n a_n^R(x),
\]
corresponding to Fact (\ref{fact:jstar}.\ref{item:jstar-density-igamma}) and
(\ref{eqn:jstar-density-exp}), where
\begin{equation}
\label{eqn:ell}
  a_n^L(x) = \pi (n+\frac{1}{2}) \Big(\frac{2}{\pi x}\Big)^{3/2} \exp \Big(- \frac{2(n+1/2)^2}{x} \Big)
\end{equation}
and
\begin{equation}
\label{eqn:rrr}
  a_n^R(x)  = \pi \big(n+\frac{1}{2} \big) \exp \Big( -\frac{(n+1/2)^2 \pi^2 x}{2} \Big).
\end{equation}
Neither $\{a_n^L(x)\}_{n=0}^\infty$ or $\{a_n^R(x)\}_{n=0}^\infty$ are
decreasing for all $x$, thus neither satisfy the partial sum criterion.
However, Devroye shows that $a_n^R(x)$ is decreasing on $I_R = [(\log 3) /
\pi^2, \infty)$ and that $a_n^L(x)$ is decreasing for $I_L = [0, 4 / \log 3]$.
These intervals overlap and hence one may pick $t$ in the intersection of these
two intervals to define the piecewise coefficient
\[
a_n(x) = 
\begin{cases}
a_n^L(x), & x \leq t \\
a_n^R(x), & x > t
\end{cases}
\]
so that $a_n(x) \geq a_{n+1}(x)$ for all $n$ and all $x \geq 0$.  Devroye finds
that $t = 2 / \pi$ is the best choice of $t$ for his $\J(1,0)$ sampler, which is
where $a_0^L(x)=a_0^R(x)$ Below we show that this still holds for $\JJ(1,z)$.
Thus the density $f$ may be written as
\[
f(x) = \sum_{i=0}^\infty (-1)^n a_n(x)
\]
and this representation does satisfy the partial sum criterion
(\ref{eqn:partial-sum-criterion}).  The density of $\J(1,z)$ is then
\[
f(x|z) = \cosh(z) \exp(-xz^2 / 2) f(x)
\]
according to our construction of $\JJ(1,z)$, in which case it also has an
infinite sum representation
\[
f(x|z) = \sum_{i=0}^\infty (-1) a_n(x|z), \; a_n(x|z) = \cosh(z) \exp(-xz^2/2) a_n(x)
\]
that satisfies (\ref{eqn:partial-sum-criterion}) for the partial sums $S_n(x|z)
= \sum_{i=0}^n (-1)^n a_i(x|z)$, as
\[
a_n(x) \geq a_{n+1}(x) \implies a_n(x|z) \geq a_{n+1}(x|z).
\]

Following our initial discussion of the von Neumann alternating sum method, all
that remains is to find a suitable proposal distribution $g$.  One would like to
find a distribution $g$ for which $\|f/g\|_\infty$ is small, since this controls
the rejection rate.  A natural candidate for $g$ is the density defined by the
kernel $S_0(x|z) = a_0(x|z)$ as $S_0(x|z) \geq f(x|z)$ for all $x$.  In that
case, we sample $X \sim g$ until $U \sim \mcU(0, a_0(x|z))$ has $U \leq f(X)$.

The proposal $g$ is thus defined from (\ref{eqn:ell}) and (\ref{eqn:rrr}) by
\[
g(x|z) \propto a_0(x|z) =
\cosh(z)
\begin{dcases}
  \Big( \frac{2}{\pi x^3} \Big)^{1/2} \exp \Big( \frac{-1}{2x} -
  \frac{z^2}{2} x \Big)
  & x < t \\
  \frac{\pi}{2} \exp \Big( - \Big[ \frac{\pi^2}{8} + \frac{z^2}{2} \Big] x \Big)
  & x \geq t.
\end{dcases}
\]
Let $a_0^L(x|z) = a_0(x|z) \one \{x < t\}$ be the left-hand kernel and define
the right-hand kernel $a_0^R(x|z)$ similarly.  Rewriting the exponent in the
left-hand kernel yields
\begin{align*}
\frac{-1}{2x} - \frac{z^2}{2}x
& = \frac{-z^2}{2x} (x^2 - 2x|z|^{-1} + 2x|z|^{-1} + z^{-2}) \\
& = \frac{-z^2}{2x} (x - |z|^{-1})^2 - |z|;
\end{align*}
hence
\[
a_0^L(x|z) = (1 + e^{-2|z|}) \; IG(x | \mu = |z|^{-1}, \lambda = 1)
\]
where $IG(x | \mu, \lambda)$ is the density of the inverse Gaussian
distribution,
\[
IG(x | \mu, \lambda) = \Big(\frac{\lambda}{2 \pi x^3}\Big)^{1/2} \exp \Big(
\frac{-\lambda (x - \mu)^2}{2 \mu^2 x} \Big).
\]
The normalizing constants 
\begin{equation}
\label{eqn:pq}
p = \int_{0}^t a_0^L(x|z) dx \text{ and } q = \int_{t}^\infty a_0^R(x|z) dx
\end{equation}
let us express $g$ as the mixture
\[
\label{eqn:proposal-mixture}
g(x | z) = \frac{q}{p+q} \frac{a_0^L(x|z)}{q} + \frac{p}{p+q} \frac{a_0^R(x|z)}{p}.
\]
and shows that (suppressing the dependence on $t$)
\begin{displaymath}
c(z) g(x|z) = a_0(x|z) \; \text{ where } \; c(z) = p(z) + q(z).
\end{displaymath}
Thus, one may draw $X \sim g(x|z)$ as
\[
X \sim
\begin{dcases}
IG(\mu = |z|^{-1}, \lambda = 1) \one \{x < t\}, & \text{ with prob. } p / (p+q)  \\
\mcE \Big(\text{rate} = \frac{\pi^2}{8} +
  \frac{z^2}{2} \Big) \one \{x \geq t\}, & \text{ with prob. } q / (p+q).
\end{dcases}
\]
One may sample from the truncated exponential by taking $X \sim Ex
\Big(\text{rate} = \frac{\pi^2}{8} + \frac{z^2}{2} \Big)$ and returning $X + t$.
Sampling from the truncated inverse Gaussian requires a bit more work (see
Appendix 2 in \cite{windle-thesis-2013}).  To recapitulate, to draw $\JJ(1,z)$:
\begin{enumerate}
\item Sample $X \sim g(x|z)$.
\item Generate $U \sim \mcU(0, a_0(X|z))$.
\item Iteratively calculate $S_n(X|z)$, starting at $S_1(X|z)$, until $U \leq
  S_n(X|z)$ for an odd $n$ or $U > S_n(X|z)$ for an even $n$.
\item Accept if $n$ is odd; return to step 1 if $n$ is even.
\end{enumerate}

\subsection{Sampling $\JJ(n, z)$}

One can use the $\J(1, z)$ sampler to generate draws from the $\JJ(n, z)$
distribution when $n$ is a positive integer.  As shown by Fact
(\ref{fact:jstar}.\ref{item:jstar-convo}), sample $X_i \sim \J(1, z)$ for $i=1,
\ldots, n$ and then return $Y = \sum_{i=1}^n X_i$.

\section{An Alternate $\JJ(h,z)$ Sampler}
\label{sec:jstar-alt}

Here we show how to sample $\JJ(h,z)$ when $h$ is not a positive integer.

\subsection{An Alternate $\JJ(h)$ sampler}

The basic strategy will be the same as in \S\ref{sec:jstar1}: find two functions
$\ell$ and $r$ such that the density $f$ is dominated by $\ell$ on $(0,t]$ and
$r$ on $(t, \infty)$.  Truncated versions of $\ell$ and $r$ can then be used to
generate a proposal.  Previously, these proposals came from the density of $f$,
which when $h=1$, has two infinite, alternating sum representations.  Pasting
together these two representations together one may immediately appeal to the
von Neumann alternating sum technique to accept or reject a proposal; but this
only works when $h=1$.  For $h \neq 1$, the density in Fact
\ref{fact:jstar}.\ref{item:jstar-density-igamma} is still valid:
\begin{equation}
\label{eqn:igamma-rep}
f(x|h) = \frac{2^h}{\Gamma(h)} \sum_{n=0}^{\infty} (-1)^n \frac{\Gamma(n+h)}{\Gamma(n+1)} 
\frac{(2n+h)}{\sqrt{2 \pi x^3}} \exp \Big( - \frac{(2n+h)^2}{2 x} \Big).
\end{equation}
We know that the coefficients of this alternating sum, which we call $a_n^L$,
are not decreasing in $n \in \bbN_0$ for all $x > 0$; they are only decreasing
in $n \in \bbN_0$ for $x$ in some interval $I_L$.  However, it is the case that
$a_n^L(x|h)$ is decreasing for sufficiently large $n$ for all $x>0$.  Thus, we
may still appeal to a von Neumann-like procedure, but only once we know that we
have reached an $n^*(x)$ so that $a_n^L(x|h)$ is decreasing for $n \geq n^*$.
The following proposition shows that we can identify when this is the case.

\begin{proposition}
  \label{prop:decreasing}
  Fix $h \geq 1$ and $x > 0$.  The coefficients $\{a_n^L(x)\}_{n=0}^\infty$ in
  (\ref{eqn:igamma-rep}) are decreasing, or they are increasing and then
  decreasing.  Further, if $a_n^L(x^*)$ is decreasing for $n \geq n^*$, then
  $a_n^L(x)$ is decreasing for $n \geq n^*$ for $x \leq x^*$.
\end{proposition}

\begin{proof}
Fix $h \geq 1$ and $x > 0$; calculate $a_{n+1}^L(x|h) / a_{n}^L(x|h)$.  It is
\begin{align*}
\frac {\Gamma(n+1)}{\Gamma(n+2)} & \frac {\Gamma(n+1+h)}{\Gamma(n+h)}
\frac {2n + 2 + h}{2n+h} \exp \Big\{ -\frac{1}{2x} \Big[ (2n+2+h)^2 - (2n+h)^2
\Big] \Big\} \\
& = \frac{n+h}{n+1} \frac{2n+h+2}{2n+h} \exp \Big\{ -\frac{1}{2x} \Big[ 4 (2n+h) +
4 \Big] \Big\} \\
& = \Big(1 + \frac{h-1}{n+1}\Big) \Big(1 + \frac{2}{2n+h}\Big) \exp \Big\{ -\frac{2}{x} \Big[ (2n+h) +
1 \Big] \Big\}.
\end{align*}
Since $x > 0$, the exponential term decays to zero as $n$ diverges and there is
smallest $n^* \in \bbN_0$ for which this quantity is less than unity.  Further,
it is less than unity for all such $n \geq n^*$ as all three terms in the
product are decreasing in $n$.  The ratio also decreases as $x$ decreases, thus
$a_n^L(y)$ is decreasing for $n \geq n^*$ when $y \leq x$.
\end{proof}

\begin{corollary}
  \label{cor:alt-alternating}
  Suppose $h \geq 1$ and $x > 0$ and let $S_n^L(x|h) = \sum_{i=0}^n (-1)^i
  a_i^L(x|h)$.  There is an $n^* \in \bbN_0$ for which $f(y|h) < S_n^L(y|h)$ for all
  even $n \geq n^*$ and $f(y|h) > S_n^L(y|h)$ for all odd $n \geq n^*$ for $y
  \leq x$.
\end{corollary}

\begin{corollary}
  \label{cor:alt-left}
  There is an $x^*(h)$,
  \[
  x^*(h) = \sup \Big\{ x : \{a_n^L(x|h)\}_{n=0}^\infty \mbox{ is decreasing } \Big\},
  \]
  so that $\{a_n^L(x|h)\}_{n=0}^\infty$ is decreasing for all $x < x^*$.  Thus
  \( \ell(x|h) = a_0^L(x|h) \) satisfies
  \[
  S_n(x|h) \leq \ell(x|h), \; \forall n \in \bbN_0, \forall x < x^*(h).
  \]
\end{corollary}

When $h=1$, we have another representation of $f(x|h)$ as an infinite
alternating sum.  This is not the case when $h \neq 1$; however, revisiting
(\ref{eqn:alternate-density}), when $h \in \bbN$, we may also write $f(x|h)$ as
\[
f(x|h) = \sum_{n=0}^\infty
\Big[ \sum_{m=1}^h \frac{A_{nm} (h-1)! }{A_{nh} (m-1)!} \frac{1}{x^{h-m}} \Big]
\frac{A_{nh} x^{h-1} e^{-c_n x}}{(h-1)!} \, , \; \; c_n =  \frac{\pi^2}{2} (n + 1/2)^2.
\]
When $x$ is large, the term with $m=h$ will dominate, leaving
\[
 \sum_{n=0}^\infty
\frac{A_{nh} x^{h-1} e^{-c_n x}}{(h-1)!} \, , \; A_{nh} = (-1)^{nh} (2 c_n)^{h/2}.
\]
Again, since $e^{-c_n x}$ decays rapidly in $n$ the first term of this sum
should be the most important.  Hence, for sufficiently large $x$, $f(x|h)$
should look like
\[
r(x|h) = \frac{A_{0h} x^{h-1} e^{-c_0 x}}{(h-1)!} 
= \frac{(\pi/2)^{h/2} x^{h-1} e^{-c_0 x}}{(h-1)!}.
\]
This will be the right hand side proposal.

\begin{conjecture}
  \label{conjecture:bound}
  The functions $\ell(x|h)$ and $r(x|h)$ dominate $f(x|h)$ on overlapping
  intervals that contain a point $t(h)$.
\end{conjecture}

For $h \geq 1$, we know that $\ell(x|h)$ will dominate $f(x|h)$ on some interval
$[0,x^*(h))$ from Corollary \ref{cor:alt-left}.  We have not proved that
$r(x|h)$ dominates $f(x|h)$ on an overlapping interval; however, we do have
numerical evidence that this is the case.  Let $\rho^L(x|h) = f(x|h) /
\ell(x|h)$ and $\rho^R(x|h) = f(x|h) / r(x|h)$.  If both $\rho^L(x|h)$ and
$\rho^R(x|h)$ are less than unity on overlapping intervals, then $\ell$ and $r$
dominate $f$ on overlapping intervals.  As seen in Figure \ref{fig:jstar-ratio},
this appears to be the case for both $\rho^L$ and $\rho^R$ on the entire real
line.  In that case, $\ell$ and $r$ are both valid bounding kernels and the
proposal density
\[
g(x|h) \propto k(x|h) = 
\begin{cases}
\ell(x|h), & x < t \\
r(x|h), & x \geq t.
\end{cases}
\]
has
\[
f(x|h) \leq k(x|h) \; \text{ for all } x > 0;
\]
further, $g(x|h)$ is a mixture
\[
g(x|h) = \frac{p}{p+q} \frac{\ell(x|h)}{p} + \frac{q}{p+q} \frac{r(x|h)}{q}
\]
where
\[
p(t|h) = \int_0^{t} \ell(x|h) dx \; \text{ and } q(t|h) = \int_{t}^\infty
r(x|h) dx
\]
and the normalizing constant of $k(x|h)$ is $c(t|h)^{-1}$ where
\[
c(t|h) = p(t|h) + q(t|h).
\]
Thus, Corollary \ref{cor:alt-alternating} and Conjecture \ref{conjecture:bound}
lead to the following sampler:
\begin{enumerate}
\item Sample $X \sim g(x|h)$
\item Sample $U \sim \mcU(0, k(X|h))$. 
\item Iteratively calculate the partial sums $S_n^L(x|h)$ until 
  \begin{itemize}
    \item $S_n^L(X|h)$ has decreased from $n-1$ to $n$, \emph{and}
    \item $U < S_n^L(X|h)$ for odd $n$ or $S_n^L(X|h) < U$ for even $n$.
    \end{itemize}
\end{enumerate}

Both $\ell(x|h)$ and $r(x|h)$ are kernels of known densities.  In particular,
\[
\ell(x|h) = \frac{2^h}{\Gamma(1)} \frac{h}{\sqrt{2 \pi}} x^{-3/2} 
\exp \Big( - \frac{h^2}{2 x} \Big),
\]
is the kernel of an inverse Gamma distribution, $\IGa(1/2, h^2/2)$, and
\[
r(x|h) = \frac{(\pi/2)^{h/2} x^{h-1} e^{- \frac{\pi^2}{8} x}}{(h-1)!}
\]
is the kernel of gamma distribution, $\Ga(h, \pi^2 / 8)$.  We can rewrite
\[
\ell(x|h) = 2^h \IGa(x | 1/2, h^2/2)
\]
to find
\[
p(t|h) = 2^h \frac{\Gamma(1/2, (h^2/2) / t)}{\Gamma(1/2)}
\]
where $\Gamma(a,b)$ is the upper incomplete gamma function, and we can rewrite
\[
r(x|h) = (4/\pi)^h \Ga(x | h, \text{rate}=\pi^2/8))
\]
to find
\[
q(t|h) = \Big(\frac{4}{\pi}\Big)^h \frac{\Gamma(h, (\pi^2/8) t)}{\Gamma(h)}.
\]
Note that this provides a way to calculate $t(h)$, since we want to minimize
$c(t|h) = p(t|h) + q(t|h)$.  This is identical to choosing the truncation point
$t(h)$ to be the point at which $\rho^L(x|h)$ and $\rho^R(x|h)$ intersect.

\begin{figure}
\centering
\includegraphics[scale=0.6]{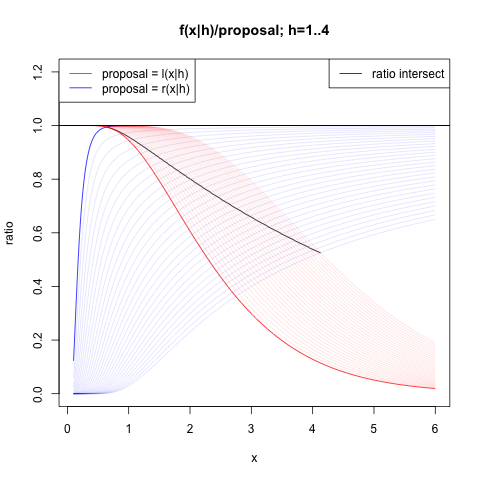}
\caption{\label{fig:jstar-ratio} A plot of the $f(x|h) / \ell(x|h)$ and $f(x|h)
  / r(x|h)$ for $h = 1.0$ to $h=4.0$ by $0.1$.  The dark lines correspond to
  $h=1$.  The curve corresponding to $\ell$ increases monotonically while the
  curve corresponding to $r$ decreases monotonically.  The black line plots the
  point of intersection between the two curves as $h$ changes.}
\end{figure}

\subsection{An Alternate $\JJ(h,z)$ Sampler}

Recall Fact \ref{fact:jstar}.\ref{item:jstar-density-igamma}, which says the
density of $\JJ(h,z)$ is
\[
f(x|h,z) = \cosh^h(z) e^{-x z^2 / 2} f(x|h)
\]
where $f(x|h)$ is given in (\ref{eqn:igamma-rep}).  Following the general path
put forth in the previous section, one finds that almost nothing changes.  In
particular, if we let $a_n^L(x|h,z) = \cosh^h(z) e^{-xz^2/2} a_n^L(x|h)$ and let
$S_n^L(x|h,z) = \sum_{i=0}^n (-1)^i a_n^L(x|h,z)$, then the analogous
propositions, corollaries, and conjectures from the previous section still hold.
In particular,
\[
\frac{a_{n+1}^L(x|h)}{a_n^L(x|h)} = \frac{a_{n+1}^L(x|h,z)}{a_n^L(x|h,z)}
\]
so Proposition \ref{prop:decreasing}, Corollary \ref{cor:alt-alternating}, and
Corollary \ref{cor:alt-left} hold with $a_n^L(x|h)$ replaced by $a_n^L(x|h,z)$,
$S_n^L(x|h)$ replaced by $S_n^L(x|h,z)$, and $\ell(x|h)$ replaced by
$\ell(x|h,z) = a_n^L(x|h,z)$.  Additionally, nothing changes with regards the
bounding kernel since
\[
f(x|h) \leq k(x|h) \iff f(x|h,z) \leq k(x|h,z) 
\]
where
\[
k(x|h,z) = \cosh^h(z) e^{-x z^2 / 2} k(x|h).
\]
Hence the only major change is the form of the proposal density and the
corresponding mixture representation.  After adjusting, the left bounding
kernel becomes
\[
\ell(x|h,z) = \cosh^h(z) 2^h \frac{h}{\sqrt{2 \pi}} x^{-3/2} 
\exp \Big( -\frac{h^2}{2x} - \frac{xz^2}{2} \Big),
\]
and the right bounding kernel becomes
\[
r(x|h,z) = \cosh^h(z) \frac{(\pi/2)^{h/2} x^{h-1}}{(h-1)!} \exp \Big[ -
  \Big(\frac{\pi^2}{8} + \frac{z^2}{2}\Big) x \Big].
\]
Let
\[
g(x|h,z) \propto k(x|h,z) = 
\begin{cases}
\ell(x|h,z), & x < t(h) \\
r(x|h,z), & x \geq t(h),
\end{cases}
\]
and
\[
p(t|h,z) = \int_{0}^t \ell(x|h,z) dx \; \text{ and } \; q(t|h,z) =
\int_{t}^\infty r(x|h,z) dx.
\]
Then one can represent $g(x|h,z)$ as the mixture
\[
g(x|h,z) = \frac{p}{p+q} \frac{\ell(x|h,z)}{p} 
+ \frac{q}{p+q} \frac{r(x|h,z)}{q} 
\]
and the normalizing constant of $k(x|h,z)$ is (suppressing the dependence on $t$)
\[
c(h,z) = p(h,z) + q(h,z)
\]
Thus, one can sample $\JJ(h,z)$ by
\begin{enumerate}
\item Sample $X \sim g(x|h,z)$
\item Sample $U \sim \mcU(0, k(x|h))$. 
\item Iteratively calculate the partial sums $S_n^L(x|h)$ until 
  \begin{itemize}
    \item $S_n^L(X|h)$ has decreased from $n-1$ to $n$, \emph{and}
    \item $U < S_n^L(X|h)$ for odd $n$ or $S_n^L(X|h) < U$ for even $n$.
    \end{itemize}
\end{enumerate}
Note that the above procedure uses $k(x|h)$ and $S_n(x|h)$ instead of
$k(x|h,z)$ and $S_n(x|h,z)$.  This is because
\[
\tilde f(x|h) / \tilde g(x|h) = \tilde f(x|h,z) / \tilde g(x|h,z)
\]
and
\[
\tilde f(x|h) / S_n^L(x|h) = \tilde f(x|h,z) / S_n^L(x|h,z).
\]

Again, the kernels $\ell(x|h,z)$ and $r(x|h,z)$ are recognizable.  The
exponential term of $\ell(x|h,z)$ is
\begin{align*}
-\frac{z^2}{2x} \Big[ \big( \frac{h}{z} \big)^2 + x^2 \Big] .
\end{align*}
Completing the square 
yields
\[
- \frac{(z/h)^2 h^2}{2x} \Big[ (x - h / z)^2 \Big] - z h;
\]
so
\[
\ell(x|h,z) = (1+e^{-2|z|})^h \frac{h}{\sqrt{2 \pi x^3}} \exp \Big( -
\frac{(z/h)^2 h^2}{2x} \Big[ (x - h / z)^2 \Big] \Big),
\]
which is the kernel of an inverse Gaussian distribution with parameters
$\mu=h/z$ and $\lambda = h^2$.  The right kernel is a gamma distribution with
shape parameter $h$ and rate parameter $\lambda_z = \pi^2 / 8 + z^2 / 2$.  Thus,
the left hand is
\[
\ell(x|h,z) = (1+e^{-2|z|})^h IG(x|\mu = h/z, \lambda=h^2) \; \text{ for } z > 0
\]
and
\[
\ell(x|h,0) = 2^h \IGa(x | 1/2, h^2/2);
\]
the right hand kernel is
\[
r(x|h,z) = \Big( \frac{\pi/2}{\lambda_z} \Big)^h \Ga(x | h, \mbox{rate}=\lambda_z), \;
\lambda_z = \pi^2/8 + z^2/2;
\]
and the respective weights are
\[
p(t|h,z) = (2^h e^{-zh}) \Phi_{IG}(t | h/z, h^2),
\]
\[
p(t|h,0) = 2^h \frac{\Gamma(1/2, (h^2/2) (1/t))}{\Gamma(1/2)},
\]
and
\[
q(t|h,z) = \Big(\frac{\pi/2}{\lambda_z}\Big)^{h} \frac{\Gamma(h, \lambda_z t)}{\Gamma(h)}.
\]

\subsubsection*{Truncation Point}

The normalizing constant $c(t|h,z)$ is
\[
c(t|h,z) = \int_0^t \cosh^h(z) e^{-xz^2/2} \ell(x|h) dx 
+ \int_{t}^\infty \cosh^h(z) e^{-xz^2 / 2} r(x|h) dx.
\]
To minimize $c(t|h,z)$ over $t$, note that the critical points, which satisfy
\[
\cosh^h(z) e^{-xz^2/2} \Big[\ell(x|h) - r(x|h)\Big] = 0,
\]
are independent of $z$.  Hence we only need to calculate the best $t=t(h)$ as a
function of $h$.

\subsection{Recapitulation}

The method put forth in this section can produce draws from $\JJ(h,z)$ for $h
\geq 1$ if Conjecture \ref{conjecture:bound} holds.  We numerically verify this
is the case for $h \in [1,4]$.  In practice, to draw $\JJ(h,z)$ when $h > 4$, we
take sums independent $\JJ$ random variates like before.  The new sampler is
limited in two ways.  First, the best truncation point $t$ is a function of $h$,
and must be calculated numerically.  Second, the normalizing constant $c(h,z)$
grows as $h$ increases.  The former is not too troubling as one may precompute
many $t(h)$ and then interpolate between values of $h$ not specified.  However,
the latter is disturbing as $1/c(h,z)$ is the probability of accepting a
proposal.  Thus, as $h$ increases the probability of accepting a proposal
decreases.  To address this deficiency, we devise yet another sampler.

\section{An Approximate $\JJ(b,z)$ Sampler}
\label{sec:saddlepoint}

\cite{daniels-1954} provides a method to construct approximations to the density
of the mean of $n$ independent and identically distributed random variables.
More generally, Daniels procedure produces approximations to the density of
$X(n) / n$ where $X(h)$ is an infinitely divisible family
\citep{sato-1999-book}.  The approximation improves as $n$ increases.  This is
precisely the scenario we are interested in addressing, as $\JJ(n,z)$ is
infinitely divisible and the two previously proposed samplers do not perform
well when sampling $\JJ(n,z)$, or equivalently $\JJ(n,z)/n$, for large $n$.

\subsection{The Saddle Point Approximation}

The method of \cite{daniels-1954} and variants thereof are known as saddlepoint
approximations or the method of steepest decent.  In addition to
\cite{daniels-1954}, \cite{murray-1974-book} provides an accessible explanation
of the asymptotic expansion and approximation, including numerous helpful
graphics.  A more technical analysis may be found in the paper by
\cite{barndorff-nielsen-cox-1979} and the books by \cite{butler-2007-book} and
\cite{jensen-1995-book}.  \cite{mcleish-2010} provides several examples of
simulating random variates following the approach of \cite{lugannani-rice-1980}.
Below, we briefly summarize the basic idea behind the approximation following
\cite{daniels-1954}.

Let $X(h)$ be an infinitely divisible family.  Let $M(t)$ denote the moment
generating function of $X(1)$, and let $K(t)$ denote its cumulant generating
function:
\[
M(t) = e^{K(t)} = \int_{-\infty}^\infty e^{tx} f(x) dx.
\]
where $f(x)$ is the density of the random variable $X(1)$.  Let $\bar x$ denote
$X(n)/n$, which can be thought of as the sample mean of $n$ independent $X(1)$
random variables when $n$ is an integer. The MGF of $\bar x$ is $M^n(t/n)$ and
its Fourier inversion is
\[
f_n(\bar x) = \frac{1}{2 \pi} \int_{-\infty}^{\infty} M^n(i t/n) e^{-i t \bar x}
d t 
=  \frac{n}{2 \pi} \int_{-\infty}^{\infty} M^n(i t) e^{-i n t \bar x}
d t 
\]
where $f_n$ is the density of $X(n)/n$.  The goal is to pick the path of this
integral in a way that concentrates as much mass as possible at a single point.
Changing variables to $T = it$ and phrasing this integral in terms of the
cumulant generating function yields
\[
f_n(\bar x) = \frac{n}{2 \pi i} \int_{-\infty i}^{\infty i} e^{n[K(T)- T \bar x]} dT.
\]
One can concentrate mass at $T_0 + 0i$ where $T_0$ is
chosen to minimize
\[
K(T) - T \bar x \; \text{ over } T \in \bbR,
\]
which will be a saddle point.  Consequently, one may descend quickly in the
directions perpendicular to the real axis at $T_0 + 0i$, which leads to an
integral like
\[
f_n(\bar x) = \frac{n}{2 \pi i} \int_{T_0 -\infty i}^{T_0 + \infty i} e^{n (
  K(T) - T \bar x)} dT,
\]
though some care must be taken with the path of integration near $T_0 + 0i$.
Performing an asymptotic expansion of $K(T)$ at $T_0$ and integrating yields
the approximation of Daniels:
\[
sp_n(\bar x) = \Big( \frac{n}{2 \pi} \Big)^{1/2} K''(T_0)^{-1/2} e^{n [K(T_0) -
  T_0 \bar x]};
\]
note $T_0(x)$ solves
\begin{equation}
\label{eqn:t0}
K'(T_0) - \bar x = 0.
\end{equation}
\cite{daniels-1954} (p.\ 639) provides conditions that ensure the approximation
will hold, which in the case of the $\JJ(1,z)$ distribution are
\[
\lim_{u \ra (\pi^2/8)^-} K_0'(u) = \infty \; \text{ and } \; \lim_{u \ra -\infty} K_0'(u) = 0
\]
where $K_0(u) = \log \cos \sqrt{2u}$ is the cumulant generating function of
$\JJ(1)$.  As seen in Fact \ref{fact:cgf-derivatives}, this is indeed the case.

\subsection{Sampling the saddlepoint approximation}
\label{sec:largeb}

The saddlepoint approximation provides a good point-wise approximation of the
density of $\JJ(n,z)/n$.  To make this useful for \Polya-Gamma data
augmentation, we need to sample from the density proportional to $sp_n(x)$.
(Henceforth we drop the bar notation for $\bar x$.)  One general approach is to
bound $\log sp_n(x)$ from above by piecewise linear functions, in which case the
approximation will consist of a mixture of truncated exponentials.  When the
log-density is a concave functions, one is assured that such an approximation
exists.  Devroye provides several examples of how this may be used in practice,
even for the case of arbitrary log-concave densities \citep{devroye-1986,
  devroye-2012}.  Figure \ref{fig:log-concave-density} shows an example of a
piecewise linear envelope that bounds a log-concave density.  One can construct
such an envelope by picking points $\{x_i\}$ on the the graph of the density
$f$, finding the tangent lines $L_i$ at each point, and then constructing the
function $e(x) = \min_{i} L_i(x)$, which corresponds to a piecewise linear
function.

\begin{figure}
\centering
\includegraphics[scale=0.5]{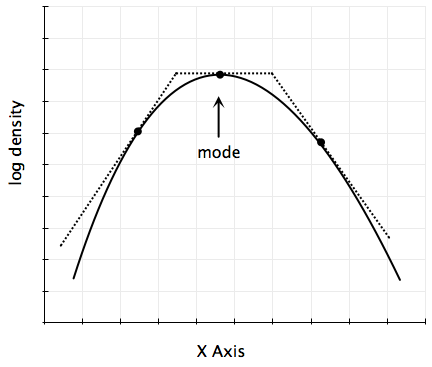}
\caption{\label{fig:log-concave-density} A log concave density bounded by a
  piecewise linear function.}
\end{figure}

We follow the piecewise linear envelope approach, though with a few
modifications.  In particular, we will bound the term $K(t) - t x$ found in the
exponent of $sp_n(x)$ rather than the kernel itself using functions more complex
than affine transforms.  It will require some care to make sure that the
subsequent envelope does not supersede $\log sp_n(x)$ too much.  However, by
working with $K(t) - tx$ directly, we avoid having to deal with the $K''(t)$
term in $sp_n(x)$, which will causes the mode of $sp_n(x)$ to shift as $n$
changes.

Recall that $t$ is implicitly a function of $x$ that arises via the minimization
of $K(t) - tx$ over $t$.  This may be phrased in terms of convex duality via
\begin{equation}
\label{eqn:concave-dual}
\phi(x) = \min_{s \in \bbR} \Big\{ K(s) - s x \Big\}
\end{equation}
where $K(t)$ is the cumulant generating function: $K(t)$ is strictly convex on
$\dom K = \{t : K(t) < \infty\}$ as $\JJ(1,z)$ has a second moment
\citep{jensen-1995-book}.  Using this notation, we may write
\[
sp_n(x) = \Big( \frac{n}{2 \pi} \Big)^{1/2} K''(t(x))^{-1/2} e^{n \phi(x)}.
\]
When needed, we will write $K_z(t)$ to denote the explicit dependence on
$z$, though usually we will suppress the dependence on $z$.  The connection to
duality will help us find a good bound for $\phi(x)$; the following facts will
be useful.

\begin{fact}
\label{fact:dual}
Let $K$ be the cumulant generating function of $\JJ(1,z)$.  Let $\phi(x)$ be the
concave dual of $K$ as in (\ref{eqn:concave-dual}).  Let
\[
t(x) = \argmin{s \in \bbR} \Big\{ K(s) - s x \Big\}.
\]
Assume that when we write $t$ we are implicitly evaluating it at $x$.  Then
\begin{enumerate}
\item $K(t)$ is strictly convex.
\item $K(t)$ is smooth.
\item $\displaystyle K'(t) = x$;
\item $\displaystyle \phi(x) = K(t) - t x$;
\item $\displaystyle \phi'(x) = -t$;
\item $\displaystyle \dd{t}{x}(x) = [K''(t)]^{-1}$;
\item As seen by item (3), $\phi'(x)$ is maximized when $t(x)=0$.  Thus, 
\[
m = \argmax{x} \phi(x) \; \text{ is attained when } \; m = K'(0).
\]
\end{enumerate}

\end{fact}

\begin{proof}
  \cite{barndorff-nielsen-1978-book} shows that (1) holds so long as $\JJ(1,z)$
  has a second moment, which it does.  The cumulant generating function $K(t) =
  - \log \cos \sqrt{2t}$ is smooth by composition of smooth functions so long as
  \[
  \cos \sqrt{2t} = 
  \begin{cases}
    \cos \sqrt{2t}, &  t \geq 0 \\
    \cosh \sqrt{2|t|}, & t < 0 
  \end{cases}
  \]
  is smooth.  For $t \neq 0$ this holds since $\cos$ and $\cosh$ are smooth and
  $\sqrt{2t}$ is smooth for $t \neq 0$.  For $t=0$, this follows from the Taylor
  expansion of $\cos$ and $\cosh$.  Items (3)-(7) are consequences of (1) and
  (2).
\end{proof}

\begin{remark}
\label{remark:different-variables}
Sometimes it will be helpful to work with a shifted version of $t$: \( u = t -
z^2/2.  \) To reiterate, we will go between three different variables: $x$, $t$,
and $u$ characterized by the bijections
\begin{enumerate}
\item $x = K'(t)$ and
\item $u = t - z^2 / 2$.
\end{enumerate}
\end{remark}

It will also be helpful to have the derivatives of $K$ on hand and a few facts
about $x$ and $u$.
\begin{fact}
\label{fact:cgf-derivatives}
Recall that $K(t) = \log \cosh(z) - \log \cos \sqrt{2u}$ is the cumulant
generating function of $\JJ(1,z)$.  Its derivatives, with respect to $t$, are:
\begin{enumerate}
\item $\displaystyle K'(t) = \utan{2u}$;
\item $\displaystyle K''(t) = \frac{\tan^2(\sqrt{2u})}{2u} + \frac{1}{2u} \Big(1
  - \utan{2u} \Big)$.
\end{enumerate}
Note that we are implicitly evaluating $u$ at $t$ as described in Remark
\ref{remark:different-variables}.  As shown above, $K'(t) = x$.  Evaluating
$K''$ at $t(x)$ yields
\[
K''(t) = x^2 + \frac{1}{2u} (1 - x).
\]
We may write $\utan{s}$ piecewise as
\[
\utan{s} = 
\begin{cases}
\utan{s}, & s > 0 \\
\utanh{|s|},  & s < 0 \\
1, & s = 0.
\end{cases}
\]
The last fact can be seen by taking the Taylor expansion around $s=0$.  Thus, $u
< 0 \iff x < 1$, $u> 0 \iff x > 1$, and $u=0 \iff x = 1$.
\end{fact}

This leads to the following two claims, which will help us bound the saddlepoint
approximation.  Notice that in each case, we adjust $\phi(x)$ to match the shape
of the tails as suggested by Remark \ref{remark:tails}.

\begin{lemma}
\label{lemma:right-eta}
The function $\eta_r(x) = \phi(x) - (\log(x) - \log(x_c))$ is strictly concave
for $x > 0$.
\end{lemma}

\begin{proof}
Taking derivatives:
\[
\eta_r'(x) = \phi'(x) - \frac{1}{x}
\]
and
\[
\eta_r''(x) = - \dd{t}{x}(x) + \frac{1}{x^2}.
\]
Using Fact \ref{fact:dual}, this is negative if and only if
\[
[K''(t)]^{-1} \geq \frac{1}{x^2} \iff x^2 \geq K''(t) \iff 0 \geq \frac{(1-x)}{2u}.
\]
When $x > 1$, $u(x) > 0$, and $\eta_r''(x) < 0$.  When $x < 1$, $u(x) < 0$, and
$\eta_r''(x) < 0$.  Continuity of $K''$ ensures that $\eta_r''(1) \leq 0$.
\end{proof}

\begin{lemma}
\label{lemma:left-eta}
  The function $\eta_l(x) = \phi(x) - 
  \frac{1}{2} \Big(\frac{1}{x_c}- \frac{1}{x}\Big)$ is strictly concave for $x >
  0$.
\end{lemma}

\begin{proof}
Taking derivatives:
\[
\eta_\ell'(x) = \phi'(x) - \frac{1}{2x^2}
\]
and
\[
\eta_\ell''(x) = - \dd{t}{x}(x) + \frac{1}{x^3}.
\]
Using Fact \ref{fact:dual}, this is negative if and only if
\[
[K''(t)]^{-1} \geq \frac{1}{x^3} \iff x^3 \geq K''(t) \iff (x^2 +
\frac{1}{2u})(x-1) \geq 0.
\]
Again, we know that when $x > 1$, $u > 0$, and hence $\eta_l(x) < 0$.  When $x <
1$ we need to show that $x^2 + 1/(2u) < 0$.  This is equivalent to showing that
\[
x^2 < -\frac{1}{2u} \iff 2u x^2 > -1, \, u < 0.
\]
That is
\[
\tan^2{\sqrt{2u}} > -1 \iff \tanh{\sqrt{|2u|}} > -1, \, \text{ for } u < 0,
\]
which indeed holds.  Thus, when $x < 1$, $\eta_l(x) < 0$.  Again, continuity of
$K''$ then ensures that $\eta_l''(1) \leq 0$.
\end{proof}

These two lemmas ensure the following claim.
\begin{lemma}
\label{lemma:eta-envelope}
Let 
\[
\delta(x)
=
\begin{cases}
\frac{1}{2} \Big(\frac{1}{x_c}- \frac{1}{x}\Big) & x \leq x_c, \\
\log(x) - \log(x_c), & x > x_c.
\end{cases}
\]
Then $\eta(x) = \phi(x) - \delta(x)$, is continuous on $\bbR$ and concave on the
intervals $(0, x_c)$ and $(x_c, \infty)$.
\end{lemma}

We may create an envelope enclosing $\phi$ in the following way.  See Figure
\ref{fig:envelope} for a graphical interpretation.

\begin{enumerate}
\item Pick three points $x_\ell < x_c < x_r$ corresponding to left, center, and
  right.
\item Find the tangent lines $L_\ell$ and $L_r$ that touch the graph of $\eta$
  at $x_\ell$ and $x_r$.
\item Construct an envelope of $\eta$ using those two lines, that is
\[
e(x) = 
\begin{cases}
L_\ell(x), & x < x_c, \\
L_r(x), & x \geq x_c.
\end{cases}
\]
\end{enumerate}
Then an envelope for $\phi(x)$ is
\[
\phi(x) \leq e(x) + \delta(x).
\]

\begin{conjecture}
  $K''(t) / x^2$ is increasing on $x > 0$ with $\lim_{x \ra 0^+} K''(t) / x^2 = 0$ and
  $\lim_{x \ra \infty} K''(t) / x^2 = 1$ and $K''(t) / x^3$ is decreasing on $x > 0$ with
  $\lim_{x \ra 0^+} K''(t) / x^3 = 1$ and $\lim_{x \ra \infty} K''(t) / x^3 = 0$.
\end{conjecture}

This can be seen by plotting these functions; however, we do not have a complete
proof currently.  Instead, we employ the following lemma.

\begin{lemma}
  \label{lemma:k2bound}
  Given $x_c \in (0, \infty)$, there are constants $\alpha_\ell, \alpha_r > 0$
  such that $K''(t)$ satisfies
  \[
  1 \geq \frac{K''(t)}{x^3} \geq \alpha_\ell \; \text{ for } x < x_c
  \]
  and
  \[
  1 \geq \frac{K''(t)}{x^2} \geq \alpha_r \; \text{ for } x > x_c.
  \]
\end{lemma}

\begin{proof}
  The upper bounds are verified in the proofs of Lemmas \ref{lemma:left-eta} and
  \ref{lemma:right-eta}.  For the lower bounds, recall that $K''(t(x)) > 0$ for
  $x \in I_M := [1/M, M]$ for any $M > 1$.  Thus, $K''(t(x))$ is bounded from
  below on $I_M$.  In addition, $x^2$ and $x^3$ are bounded on the same interval
  from above.  Hence the ratios $K''(t) / x^3$ and $K''(t) / x^2$ are bounded
  from below on $I_M$ and we only need to consider the tail behavior of these
  ratios.

  Let $v(x) = 2u(x)$.  When $x < 1$, $v < 0$, and \( x^2 |v| = \tanh^2
  \sqrt{|v|}\) the ratio
  \[
  K''(t) / x^3 = 
  \frac{1}{x} - \frac{1-x}{x (x^2|v|)} = \frac{1}{x} - \frac{1-x}{x \tanh^2\sqrt{|v|}}.
  \]
  Employing the trigonometric identity $-\sinh^2 = 1 - \coth^2$ and writing out
  $x(v)$ yields
  \[
  \frac{1}{\tanh^2 \sqrt{|v|}} + \frac{1}{x} \Big(1 - \coth^2 \sqrt{|v|} \Big)
  =
  \frac{1}{\tanh^2 \sqrt{|v|}} - \frac{\sqrt{|v|} \cosh \sqrt{|v|}}{\sinh^3 \sqrt{|v|}}.
  \]
  As $v \ra - \infty$ the first term converges to unity
  while the second term vanishes.  Since $v$ is an increasing function of $x$
  that diverges to $-\infty$ as $x \ra 0^+$, for any $1 > \alpha_\ell > 0$,
  there is an $M > 1$ such that $K''(t)/x^3 > \alpha_\ell$ for $x < 1/M$.
  
  Similarly, when $x > 1$, $v > 0$, and $x^2 v = \tan^2 \sqrt{v}$ the ratio
  \[
  K''(t) / x^2 = 1 + \frac{1-x}{x^2 v} = 1 + \frac{1-x}{\tan^2 \sqrt{v}}.
  \]
  The last term can be rewritten as
  \[
  \frac{1-x}{\tan^2 \sqrt{v}} = \frac{1}{\tan \sqrt{v}} \Big( \frac{1}{\tan
    \sqrt{v}} - \frac{1}{\sqrt{v}} \Big),
  \]
  which converges to zero as $v \ra {(\pi/2)^2}^{-}$.  Since $v$ is increasing
  in $x$ and converges to $(\pi/2)^2$ as $x \ra \infty$, for any $1 > \alpha_r >
  0$, there is an $M > 1$ such that $K''(t) / x^2 > \alpha_r$ for $x > M$.

\end{proof}

Lemma \ref{lemma:eta-envelope} and Lemma \ref{lemma:k2bound} give us the
following proposition.
\begin{proposition}
\label{prop:saddlepoint-proposal}
  There exists constants $1 > \alpha_\ell, \alpha_r > 0$ such that the saddle
  point approximation of $\JJ(n,z)/n$ is bounded by the envelope
  \[
  k(x|h,z) = \Big( \frac{n}{2 \pi} \Big)^{1/2}
  \begin{cases}
    \alpha_\ell^{-1/2}  e^{\frac{n}{2x_c}} \; x^{-3/2}
      \exp \Big( - \frac{n}{2x} + n L_\ell(x|z) \Big), & x < x_c \\
      \alpha_r^{-1/2}  x_c^n \; x^{n-1}
      \exp \Big( n L_r(x|z) \Big), & x > x_c,
    \end{cases}
  \]
  where $L_\ell$ is the line touching $\eta$ at $x_\ell$ and $L_r$ is the line
  touching $\eta$ at $x_r$.  Further, $L_\ell'$ and $L_r'$ are negative when
  $x_\ell \geq m = \argmax{x} \phi(x)$.
\end{proposition}

\begin{proof}
  Lemma \ref{lemma:eta-envelope} and Lemma \ref{lemma:k2bound} provide the
  envelope.  It only remains to show that the slopes of $L_\ell$ and $L_r$ are
  negative when $x_\ell \geq m$.  Note that the concavity of $\phi$ ensures that
  $\phi'(x) \leq 0$ when $x \geq m$.  Thus, in the left case, $L_\ell'(x_\ell) =
  \phi'(x_\ell) - \frac{1}{2x_{\ell}^2} < 0.$ Similarly, in the right case,
  $L_r'(x_r) = \phi'(x_r) - \frac{1}{x_r} < 0$.
\end{proof}

Given the stipulation that $x_\ell \geq \argmax{x} {\phi(x)}$, the left hand
kernel, $k_\ell(x|h,z)$, is an inverse Gaussian kernel while the right hand
kernel, $k_r(x|h,z)$, is a gamma kernel.  To see this let $\rho_\ell = - 2
L_\ell'(x)$ and $b_\ell = L_\ell(0)$; then the exponent of the left hand kernel
is
\[
n b_\ell - \frac{n \rho_\ell x}{2} - \frac{n}{2x} = \frac{-n \rho_\ell}{2 x} \Big(
\frac{1}{\rho_\ell} + x^2 \Big) + n b_\ell.
\]
Taking the first term and completing the square yields
\[
\frac{-n \rho_\ell}{2x} \Big(x - \frac{1}{\sqrt{\rho_\ell}}\Big)^2 - n \sqrt{\rho_\ell}.
\]
Thus 
\[
k_\ell(x|h,z) = \kappa_\ell \Big( \frac{n}{2 \pi x^3} \Big)^{1/2} \exp \Big\{
\frac{-n \rho_\ell}{2 x} \Big( x - \frac{1}{\sqrt{\rho_\ell}} \Big)^2 \Big\}
\]
where
\[
\kappa_\ell = \alpha_\ell^{-1/2} e^{\frac{n}{2 x_c} + n b_\ell - n \sqrt{\rho_\ell}}
\]
so $k_\ell$ is the kernel of an inverse Gaussian distribution with parameters
$\mu = 1/\sqrt{\rho_\ell}$ and $\lambda = n$.  For the right hand kernel let
$\rho_r = -L_r'(x)$ and $b_r = L_r(0)$, which yields
\[
k_r(x|h,z) = \kappa_r \frac{(n \rho_r)^n x^{n-1}}{\Gamma(n)} e^{-n \rho_r x}
\]
where
\[
\kappa_r = \Big(\frac{n}{2 \pi \alpha_r}\Big)^{1/2} \frac{e^{n b_r} \Gamma(n)}{(n
  \rho_r)^n}
\]
so $k_r$ is the kernel of a Gamma distribution with shape $n$ and rate $n
\rho_r$.  These two observations show that $g(x|h,z) \propto k(x|h,z)$ is a
mixture, which can be sampled in a manner similar to the previous two
algorithms.

We have yet to specify the points $x_\ell$, $x_c$, or $x_r$.  As mentioned at
the outset, it is important to choose these points carefully so that the
envelope does not exceed the target density by too much.  Currently, we set
$x_\ell$ to be the mode of $\phi$.  By picking $x_\ell$ to match the maximum of
$\phi$ we guarantee that the mode of $sp_n(x)$ matches the mode of $k(x|h,z)$ as
$n \ra \infty$.  We could set $x_r = 1.2 x_\ell$ and then chose $x_c$ so that
$L_\ell(x_c) = L_r(x_c)$, in which case the envelope $e$ is continuous.  When
that is the case the following proposition holds.  However, this requires a
non-linear solve, so in practice we simply set $x_c = 1.1 x_\ell$.

\begin{proposition}
  Suppose $e$ is continuous.  Let $m$ be the maximum of $\phi(x)$.  If $x_\ell =
  m$, then the envelope $e(x) + \delta(x)$ takes on its maximum at $m$ as well.
  Further, as $n \ra \infty$, the mode of the saddlepoint approximation
  converges to the mode of $k(x|h,z)$.
\end{proposition}

\begin{proof}
  Suppose $m$ maximizes $\phi$ and $x_\ell = m$.  Then
\[
e'(x_\ell) + \delta'(x_\ell) = \phi'(x_\ell) = 0.
\]
Since $e'(x_\ell) + \delta'(x_\ell)$ is strictly concave on $(0, x_c]$, $x_\ell$
must be the maximum of the left-hand portion of the envelope for $\phi$.  We
will show that this is the only maximum by contradiction.  Suppose the
right-hand portion of the envelope of $\phi$ has a maximum at $y > x_c$.  Since
that portion is also strictly concave, we must have $\phi'(y) - \delta'(y) = 0
\implies \phi'(y) = \delta'(y)$.  But $\phi'(y) < 0$ since $y > m$ and
$\delta'(y) = 1/y > 0$, a contradiction.

To see that the modes of $sp_n$ and $k(x|h,z)$ converge as $n \ra \infty$, take
the log of each.  The log of the saddlepoint approximation is like
\[
\phi(x) - \frac{1}{2n} \log K''(t(x))
\]
while the log of the left hand kernel, where the maximum is, is like
\[
e(x) + \delta(x) - \frac{3}{2n} \log x.
\]
Since $\delta$ and $\phi$ are concave and decay faster than $\log x$ as $x \ra
0^+$ and $\log x$ is increasing, we know that the argmax of each converges to
$m$.

%
\end{proof}

Collecting all of the above lemmas leads to the following approximate sampler of
$\JJ(n,z)$.  Some preliminary notation: let $\phi_z(x)$ be the concave dual of
$K_z(t)$; let $sp_n(x|z)$ be the saddle point approximation; and let $m$ be the
mode of $\phi_z$: $m = (\tanh{z}) / z$.
\begin{itemize}
\item Preprocess.
\begin{enumerate}
\item Let $x_\ell = m$, $x_c = 1.1 x_\ell$, and $x_r = 1.2 x_\ell$.
\item Calculate the tangent lines of $\eta$ at $x_\ell$ and $x_r$; $L_\ell(x|z)$
  and $L_r(x|z)$ respectively.
\item Construct the proposal $g(x|n,z) \propto k(x|n,z)$.
\end{enumerate}
\item Accept/reject.
\begin{enumerate}
\item Draw $X \sim g(x|n,z)$.
\item Draw $U \sim \mcU(0, k(X|n,z)$.
\item If $U > sp_n(X|z)$, return to 1.
\item Return $nX$.
\end{enumerate}
\end{itemize}

\subsection{Recapitulation}

The saddlepoint approximation sampler generates approximate $\JJ(n,z)$ random
variates when $n$ is large, a regime that the previous two samplers handled
poorly.  The saddlepoint approximation sampler is similar to the previous two
samplers in that the proposal is a mixture of an inverse Gaussian kernel and a
gamma kernel.  Hence the basic framework to simulate the approximation requires
routines already developed in \S\ref{sec:jstar1} and \S\ref{sec:jstar-alt}.  We
have identified that a good choice of $x_\ell$ is the mode of $\phi$; however,
we have not yet identified the optimal choices of $x_c$ and $x_r$.  The values
of $x_\ell$, $x_c$, and $x_r$ depend on the tilting parameter $z$, but not the
shape parameter $n$ in $\JJ(n,z)$.  Thus, one could preprocess $x_\ell$, $x_r$,
and $x_c$ for various values of $z$ and then interpolate.

\begin{figure}
\centering
\includegraphics[scale=0.5]{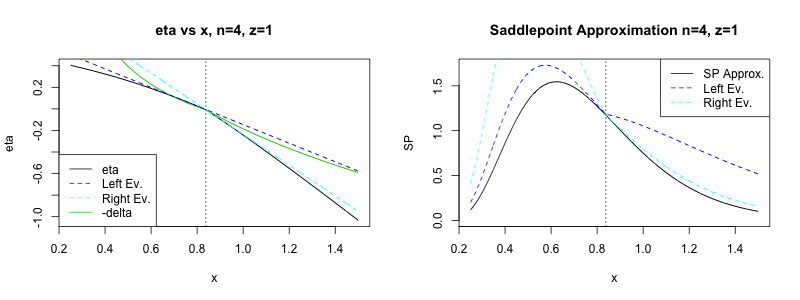}
\caption{\label{fig:envelope} The saddlepoint approximation.  The saddle point
  approximation is proportional to $[K''(t(x))]^{-0.5} \exp( n \phi(x) )$.  In
  the left plot, $\eta(x)$ is a solid black curve, which is bounded from above
  by an envelope of the dotted blue line on the left and the dotted cyan line on
  the right.  The green line is -$\delta(x)$.  On the right, the saddlepoint
  approximation in black, and the left and right envelopes are in blue and cyan
  respectively.  This bound is a bit exaggerated since $n=4$, which is rather
  small.  The bounding envelope improves as $n$ increases.}
\end{figure}



\section{Comparing the Samplers}
\label{sec:pg-compare}

We have a total of four $\JJ(n,z)$ samplers available: the method from
\S\ref{sec:jstar1}, which we call the Devroye approach, based upon sampling
$\JJ(1,z)$ random variates; the method from \S\ref{sec:jstar-alt}, which we call
the alternate approach, that lets one directly draw $\JJ(n,z)$ for $n \in
[1,4]$; the method from \S\ref{sec:saddlepoint} using the saddlepoint
approximation; and the method based upon Fact
\ref{fact:jstar}.\ref{item:sum-of-gammas}, where one simply truncates the
infinite sum after, for instance, drawing 200 gamma random variables.  Recall
that to sample $\JJ(n,z)$ using the $\JJ(1,z)$ sampler, one sums $n$ independent
copies of $\JJ(1,z)$.  Similarly, to sample $\JJ(n,z)$ when $n > 4$ using the
alternate method, we sum an appropriate number of $\JJ(b_i, z)$, $b_i \in (1,4)$
so that $\sum_{i=1}^{m} b_i = n$.

We compare these methods empirically on a MacBook Pro with 2 GHz Intel Core i7
CPU and 8GB 1333 MHz DDR3 RAM.  For a variety of $(n,z)$ pairs, we record the
time taken to sample 10,000 $\JJ(n,z)$ random variates.  Table \ref{tab:speedup}
reports the best method for each $(n,z)$ pair, along with the speed up over the
Devroye approach as measured by the ratio of the time taken to draw samples
using the Devroye method to the time taken to draw samples using the best
method.  The Devroye approach works well for $n=1,2$ while the alternate method
works well for $n=3,\ldots,10$.  The saddlepoint approximation works well for
moderate to large $n$.  These general observations do not change drastically
across different $z$, though changing $z$ can change the best sampler for fixed
$n$.  Based upon these observations, we may generate a hybrid sampler, which
uses the Devroye method when $n=1,2$, the alternate method for $n \in (1,13)
\backslash \{1,2\}$, the saddlepoint method when $13 \leq n \leq 170$, and a
normal approximation for $n \geq 170$.  The normal approximation is not strictly
necessary for large $n$, but the pre-built routines used to calculate the gamma
function break down for $n \geq 170$.  In this case, a simple fix is to
calculate the mean and variance of the $\PG(n,z)$ distribution using the moment
generating function from Fact \ref{fact:jstar}, and then sample from a normal
distribution by matching moments.  The central limit theorem suggests that this
is a reasonable approximation when $n$ is sufficiently large.



\begin{table}
\small
\begin{tabular}{l | c c c c c c }
& \multicolumn{6}{c}{Best Method} \\
$n$ $\backslash$ $z$ & 0 & 0.1 & 0.5 & 1 & 2 & 10 \\
\hline
1 & DV & DV & DV & DV & DV & DV \\
2 & DV & DV & AL & AL & AL & AL \\
3 & DV & AL & AL & AL & AL & AL \\
4 & AL & AL & AL & AL & AL & AL \\
10 & SP & AL & AL & AL & AL & AL \\
12 & SP & SP & SP & AL & AL & AL \\
14 & SP & SP & SP & SP & SP & AL \\
16 & SP & SP & SP & SP & SP & AL \\
18 & SP & SP & SP & SP & SP & SP \\
20 & SP & SP & SP & SP & SP & SP \\
30 & SP & SP & SP & SP & SP & SP \\
40 & SP & SP & SP & SP & SP & SP \\
50 & SP & SP & SP & SP & SP & SP \\
100 & SP & SP & SP & SP & SP & SP 
\end{tabular}
\begin{tabular}{| c c c c c c}
\multicolumn{6}{c}{Speed-up over $\JJ(1,z)$ sampler} \\
 0 & 0.1 & 0.5 & 1 & 2 & 10 \\
\hline
 1 & 1 & 1 & 1 & 1 & 1\\
 1 & 1 & 1 & 1.08 & 1.08 & 1.22\\
 1 & 1.26 & 1.25 & 1.29 & 1.64 & 1.78\\
 1.21 & 1.5 & 1.58 & 1.47 & 1.93 & 2.75\\
 1.34 & 1.36 & 1.3 & 1.35 & 1.7 & 2.14\\
 1.64 & 1.54 & 1.54 & 1.52 & 1.94 & 2.56\\
 1.86 & 1.72 & 1.77 & 1.7 & 1.92 & 2.26\\
 2.06 & 1.87 & 2 & 1.93 & 2.21 & 2.57\\
 2.27 & 2.07 & 2.17 & 2.15 & 2.46 & 2.42\\
 2.51 & 2.25 & 2.35 & 2.36 & 2.69 & 2.74\\
 3.68 & 3.36 & 3.57 & 3.36 & 3.92 & 4.05\\
 4.68 & 4.41 & 4.57 & 4.48 & 4.99 & 5.51\\
 5.83 & 5.16 & 5.55 & 5.55 & 6.11 & 6.78\\
 11.07 & 10.4 & 10.66 & 10.44 & 12.22 & 10.45
\end{tabular}
\caption{\label{tab:speedup} $\JJ(n,z)$ benchmarks.
  For each method and each $(n,z)$ pair the time taken to draw 10,000 samples was
  recorded and compared.  The left portion of the table lists the best method for
  each $(n,z)$ pair.  The methods benchmarked include DV, the method from
  \S\ref{sec:jstar1}; AL, the method from \S\ref{sec:jstar-alt}; SP, the method
  from \S\ref{sec:saddlepoint}; and GA, an approximate draw using a truncated sum
  of 200 gamma random variates based upon Fact
  \ref{fact:jstar}.\ref{item:sum-of-gammas}.  Notice that the truncated sum method
  never wins.  The DV method wins for small $n$;
  the AL method wins for modest $n$, and the SP method wins for medium and large
  $n$.  The right hand portion of the table shows the ratio of the time taken to
  sample each $(n,z)$ pair using DV to the time taken to sample using the best
  method.
}
\end{table}

\bibliographystyle{abbrvnat}
\bibliography{bayeslogit,/Users/jwindle/Projects/RV-Project/Notes/rk}{}

\end{document}